\documentclass{lmcs} 


\ACMCCS{[{\bf Theory of computation}]: Formal languages and automata theory; Modal and temporal logics;}

\usepackage{tikz}
\usepackage{hyperref}
\usepackage[utf8]{inputenc}
\theoremstyle{plain}\newtheorem{satz}[thm]{Satz} 

\def\eg{{\em e.g.}}
\def\cf{{\em cf.}}

\begin{document}

\title[Albegraic characterization of Forest Logics]{Algebraic characterization of Forest Logics}

\author[K.~Gelle]{Kitti Gelle}	
\address{University of Szeged, Szeged, Hungary}	
\email{kgelle@inf.u-szeged.hu}  
\thanks{This research was supported by NKFI grant no. 108448.}	

\author[Sz.~Iván]{Szabolcs Iván}	
\address{University of Szeged, Szeged, Hungary}	
\email{szabivan@inf.u-szeged.hu}  





\begin{abstract}
  \noindent
  In this paper we define future-time branching temporal logics evaluated over forests,
  that is, ordered tuples of ordered, but unranked, finite trees.
  We associate a rich class FL[$\mathcal{L}$] of temporal logics to each set L of (regular) modalities.
  Then, we define an algebraic product operation which we call the Moore product,
  which operates on forest automata, algebraic devices recognizing forest languages.
  We show a lattice isomorphism between the pseudovarieties of finite forest automata,
  closed under the Moore product, and the classes of languages of the form FL[$\mathcal{L}$].
  We demonstrate the usefulness of the algebraic approach by showing the decidability
  of the membership problem of a specific pseudovariety of finite forest automata,
  implying the decidability of the definability problem of the FL[EF] fragment of the logic
  CTL.
  Then, using the same approach, we also formulate a conjecture regarding a
  decidable characterization of the FL[AF] fragment which has currently an
  unknown decidability status (also in the setting of ranked trees).
\end{abstract}

\maketitle

In~\cite{DBLP:journals/tcs/Esik06a}, a temporal logic $\mathrm{FTL}(\mathcal{L})$
was associated with a class $\mathcal{L}$ of \emph{tree} languages.
In that setting, the structures over which the formulas were evaluated were \emph{trees}:
well-formed terms over a ranked alphabet.
The widely studied temporal logic $\mathrm{CTL}$ is also of the form $\mathrm{FTL}(\mathcal{L})$
for some suitable (finite) language class $\mathcal{L}$, thus an (algebraic, say) characterization
of these logics provides a characterization for this logic as well, and several fragments and
extensions of it are also handled in a uniform way.
In~\cite{DBLP:journals/tcs/Esik06a}, such an algebraic characterization was proved,
namely when the logic $\mathrm{FTL}(\mathcal{L})$ is expressive enough (that is, if
the so-called $\mathsf{next}$ modalities, with $\mathrm{X}_i\varphi$ meaning that
the $i$-th child of the root node of the tree satisfies $\varphi$, are expressible
in the logic in question). In that case (if an additional natural property of $\mathcal{L}$
is satisfied), an Eilenberg-type correspondence was shown between the lattice of
these language classes $\mathbf{FTL}(\mathcal{L})$ and \emph{pseudovarieties} of finite
tree automata closed under the so-called \emph{cascade product}.
Note that the decidability status of the definability problem of $\mathrm{CTL}$
(that is, to determine whether a regular tree language is definable in this logic)
is still open after some thirty years, and in the case of \emph{words}, many
logics' definability problem was shown to be decidable using algebraic methods
of this form: first one shows that a language is definable in some logic
if and only if the minimal automaton of the language (or its syntactic monoid) is contained in 
a specific pseudovariety of finite automata (or finite monoids), which is then in turn showed to
have a decidable membership problem. Notable instances of this line of reasoning
include the case of \emph{first-order} logic (which also has an unknown decidability
status for the case of trees)~\cite{McNaughton:1971:CA:1097043,schutzenberger}.
For a comprehensive treatment of logics on words see~\cite{Straubing:1994:FAF:174703}.

Extending the initial results of~\cite{DBLP:journals/tcs/Esik06a}, 
in~\cite{DBLP:journals/fuin/EsikI08} a more restricted product operation named the \emph{Moore product} of
tree automata (being a special case of the cascade product) was defined and applying this
product, we succeeded to prove an algebraic characterization of the logics
$\mathrm{FTL}(\mathcal{L})$, without the requirement on the $\mathsf{next}$ modalities:
namely, a (regular) tree language is definable in $\mathrm{FTL}(\mathcal{L})$ if and
only if its minimal automaton is contained in the least pseudovariety of tree automata
which contains all the minimal automata of the members of $\mathcal{L}$
and which is closed under the Moore product. In~\cite{DBLP:journals/fuin/EsikI08a}
the usefulness of this characterization was demonstrated by showing that
the fragment of $\mathrm{CTL}$ in which one is allowed to use only the non-strict version of
the $\mathrm{EF}$ modality, which we might call $\mathrm{TL}(\mathrm{EF}^*)$, has
a (low-degree) polynomial-time decidable definability problem (since the corresponding
pseudovariety of finite tree automata has an efficiently decidable membership problem).
(For the same result, proven by an Ehrenfeucht-Fra\"\i ss\'e type approach,
see also~\cite{Wu:2007:NCT:1225969.1226362}.)

Nowadays, instead of strictly ranked trees, unranked trees or \emph{forests} (that is,
finite tuples of finite unranked trees) are considered as models, partially due to the
larger class of real-life problems that can be modeled by them.
For example, running jQuery or XPath queries on JSON objects or XML files one usually
works with \emph{unranked} trees, hence the notion of forest is clearly a motivated
one. Now for setting forests instead of trees as the primary category of objects is
a matter of personal taste, and doing so makes the mathematic treatment more uniform.

In~\cite{DBLP:journals/corr/abs-1208-6172}, a rich class of \emph{forest logics},
called $\mathrm{TL}(\mathcal{L})$ ($\mathrm{TL}$ for ``temporal logic'') was associated
to a class $\mathcal{L}$ of modalities (analogously, but not exactly corresponding
to the logics in~\cite{DBLP:journals/tcs/Esik06a}).
There, an Eilenberg-type correspondence has been shown between the classes of languages
definable in $\mathrm{TL}(\mathcal{L})$ and pseudovarieties of \emph{forest algebra}
(which can be seen as algebraic devices extending the notion of \emph{syntactic semigroups}
from the word setting) closed under the \emph{wreath product}.
Note that characterizations of the form ``this logical construct corresponds to that
algebraic product'' are frequent, e.g., for trees a quite similar characterization,
the block product of preclones (which are also extensions of syntactic semigroups,
in this case for ranked trees) corresponds to so-called \emph{Lindström quantifiers}
(which are essentially the same constructs as the one used in $\mathrm{FTL}(\mathcal{L})$)
also exists~\cite{esik:hal-00173125}.

In this paper we propose another class of forest logics, which we call $\mathrm{FL}(\mathcal{L})$,
associated to a class $\mathcal{L}$ of modalities, which \emph{syntactically} coincides
with the $\mathrm{TL}(\mathcal{L})$ of~\cite{DBLP:journals/corr/abs-1208-6172}
(perhaps unsurprisingly, since~\cite{DBLP:journals/corr/abs-1208-6172} explicitly
states that ``This is similar to notions introduced by Esik in~\cite{DBLP:journals/tcs/Esik06a}'').
The \emph{semantics} of $\mathrm{TL}$ and $\mathrm{FL}$ differ, though, when it comes to
evaluate modalities. Specifically, in~\cite{DBLP:journals/corr/abs-1208-6172}, a tree
$a(s)$ ``tree-satisfies'' a forest formula $\varphi$ if the forest $s$ satisfies $\varphi$,
that is, there are two satisfaction relations between trees and forest formulas,
$\models$ and $\models_f$, and these relations actually differ.
In the semantics proposed in the current paper there is only one satisfaction relation.
In our view, formulas of $\mathrm{TL}$ are evaluated as they would contain a ``built-in''
$\mathsf{next}$ operator: a tree satisfies a forest formula iff the forest formed by its direct
subtrees satisfies it -- this behaviour can be modeled in $\mathrm{FL}$ be using an
explicit $\mathsf{next}$ operator first, and then the modality in question.
Thus, results of~\cite{DBLP:journals/corr/abs-1208-6172} correspond to the results
of~\cite{DBLP:journals/tcs/Esik06a}, that is, assuming the presence of the $\mathsf{next}$
modalities (reformulated the results from the tree setting to the forest setting),
while the current results lift the results of~\cite{DBLP:journals/fuin/EsikI08}
to the forest setting, providing a generalization of both~\cite{DBLP:journals/fuin/EsikI08a}
and~\cite{DBLP:journals/corr/abs-1208-6172} at the same time.

Also, in the current paper we show the applicability of our framework (in which we
work with pseudovarieties of forest \emph{automata} instead of forest \emph{algebras})
by showing that the non-strict $\mathrm{EF}$ logic has a decidable definability problem
also in this setting.
Since the class of minimal forest algebra of languages definable in the non-strict
$\mathrm{EF}$ is \emph{not} closed under taking wreath product (basically due to the
fact that the equation $aax=ax$, which holds in the minimal automaton
of the corresponding modality, is not preserved), this result
cannot be gathered via the wreath product since the logic in question falls outside
of the scope of~\cite{DBLP:journals/corr/abs-1208-6172}.
This result generalizes~\cite{DBLP:journals/fuin/EsikI08a} from trees to forests.
We think that for this result,
the (decidable) equational description of the corresponding pseudovariety
of finite forest automata is also compact and nice, and the proofs are somewhat less
heavy on technicalities in the forest setting than in the tree setting.

It is of course clear that our Moore product of automata, viewed at the level of syntactic
forest algebras, translates into a \emph{restricted} form of the wreath product of~\cite{DBLP:journals/corr/abs-1208-6172}. This is similar to the relation of
the cascade, Glushkov or Moore product of automata, which translate to restricted variants
of the wreath, block or semidirect product of the syntactic monoid.
Also, as the minimal forest automaton can be exponentially more succint than
the syntactic forest algebra, we can hope for better time complexity results when
the pseudovariety in question is shown to have a decidable membership problem.

We also note that though the strict $\mathrm{EF}$ is indeed a more expressive fragment of CTL than the non-strict variant (since non-strict $\mathrm{EF}\varphi$ can be expressed as
$\varphi\vee\mathrm{EF}\varphi$ with the strict variant of the modality),
but the logic being less expressive does not entail neither
uninterest nor having an easier definability problem.
(Just as first-order logic is less expressive than monadic
second-order logic and its definability problem still has
an unknown decidability status.) Also, the non-strict
$\mathrm{EF}$ is one of the ``simplest'' logics (as the corresponding
Moore pseudovariety of forest automata is generated by a single
two-state automaton over a binary alphabet) which falls outside
the scope of the logics of the form $\mathrm{TL}(\mathcal{L})$.

\section{Notation}
\subsection{Trees, forests}
The notions of trees, forests, contexts and other structures are following~\cite{DBLP:journals/corr/abs-1208-6172} apart from slight notational
changes (e.g., we use $\Sigma$ for the alphabet instead of $A$).

Let $\Sigma$ be a nonempty finite set (an \emph{alphabet}).
The sets $T_\Sigma$ of \emph{trees} and $F_\Sigma$ of \emph{forests} over $\Sigma$ are defined via mutual induction as the least sets satisfying the following conditions:
if $s\in F_\Sigma$ is a forest and $a\in\Sigma$ is a symbol, then $a(s)$ is a tree,
and if $n\geq 0$ is an integer and $t_1,\ldots,t_n$ are trees, then the formal (ordered) sum
$t_1+\ldots+t_n$ is a forest. In particular, for $n=0$ the \emph{empty forest}
$\mathbf{0}$ is always a forest, thus $a(\mathbf{0})$ are trees for any symbol $a\in\Sigma$.
A \emph{forest language} (over $\Sigma$) is an arbitrary set $L\subseteq F_\Sigma$ of
($\Sigma$-)forests.

\begin{exa}
	\label{example-forest}
	In our examples, we remove $\mathbf{0}$ from nonempty forests for better readability and write $a$ for $a(\mathbf{0})$,
	$a+ab$ for $a(\mathbf{0})+a(b(\mathbf{0}))$ etc. When $\Sigma=\{a,b,c,d\}$, then the following Figure
	depicts the forest $d(b(a)+a(d+a+b))+c\in F_\Sigma$:
	
	
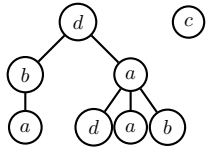
{\centering\begin{tikzpicture}[scale=0.7,thick,every node/.style={transform shape}]
		\node[circle] (root) at (0,0) [draw] {$d$};
		\node[circle] (root2) at (2.1,0) [draw] {$c$};
		\node[circle] (node1) at (-1,-1) [draw] {$b$};
		\node[circle] (node2) at (1,-1) [draw] {$a$};
		\node[circle] (node11) at (-1,-2) [draw] {$a$};
		\node[circle] (node21) at (0.3,-2) [draw] {$d$};
		\node[circle] (node22) at (1,-2) [draw] {$a$};
		\node[circle] (node23) at (1.7,-2) [draw] {$b$};
		\path (root) edge (node1);
		\path (root) edge (node2);
		\path (node1) edge (node11);
		\path (node2) edge (node21);
		\path (node2) edge (node22);
		\path (node2) edge (node23);
		\end{tikzpicture}
		
	}
	
	The forest above is a sum of two trees, one of them being $d(b(a)+a(d+a+b)$, the other being $c$.
\end{exa}

\subsection{Forest automata}

There are various algebraic devices (``automata'') recognizing forest languages.
One of them are the \emph{forest algebras} of~\cite{DBLP:journals/corr/abs-1208-6172}
another are the \emph{forest automata} of~\cite{DBLP:conf/birthday/BojanczykW08}.
For the aims of this paper, we find forest automata to be more suitable.
The reason for this is that it will be convenient to deal with the
actions induced by \emph{elementary} contexts $a(\Box)$ with $a\in\Sigma$,
and in forest algebras (which resemble closely the syntactic monoids well-known
from the case of finite words, both in the horizontal and the vertical direction)
one actually has to use pairs of forest algebras and morphisms from the free forest
algebra to these forest algebra as in our setting, the classes are not necessarily
closed under inverse morphisms. 
Using forest automata, we need only this (conceptually more simpler, we would argue) model
of computation.

A (finite) \emph{forest automaton} (over $\Sigma$) is a system
$A=(Q,\Sigma,+,0,\cdot)$ where $(Q,+,0)$ is a (finite) monoid (also called the
\emph{horizontal monoid} of $A$) and $\cdot:\Sigma\times Q\to Q$
defines a left action of $\Sigma^*$ on $Q$ (i.e. $(Q,\Sigma,\cdot)$ is a $\Sigma$-automaton).
Given the forest automaton $A$, trees $t\in T_\Sigma$ and forests $s\in F_\Sigma$
are \emph{evaluated} in $A$ to $t^A,s^A\in Q$ by structural induction as follows:
the value of a tree $t=a(s)$ is $t^A=a\cdot s^A$ and
the value of a forest $s=t_1+\ldots+t_n$ is $s^A=t_1^A+\ldots+t_n^A$.
In particular, $\mathbf{0}^A=0$, the zero of the horizontal monoid.


When the above automaton $A$ is also equipped with a set $F\subseteq Q$ of final states,
then $A$ \emph{recognizes} the forest language $L(A,F)=\{s\in F_\Sigma:s^A\in F\}$
by the set $F$ of final states. Forest languages of the form $L(A,F)$ are said to be
\emph{recognizable} in $A$, and a forest language is called \emph{recognizable}
if it is recognizable in some finite forest automaton.

Observe that $F_\Sigma$ equipped with the sum $(t_1+\ldots+t_n)+(t'_1+\ldots+t'_m)=(t_1+\ldots+t_n+t'_1+\ldots+t'_m)$ and $a\cdot s=a(s)$ (viewed as a forest consisting of a single tree) is a forest automaton.

\begin{exa}
	\label{example-ef-automaton}
	Let $\mathrm{EF}$ be the forest automaton $(\{0,1\},\{0,1\},\vee,0,\vee)$ over the alphabet $\Sigma=\{0,1\}$.
	(Note that since $\Sigma=Q$, both the action and the horizontal operation become $Q^2\to Q$ functions.)
	Then, for any forest $s\in F_\Sigma$ we have $s^{\mathrm{EF}}=1$ if and only if $s$ has a
	node (either root or non-root) labeled $1$.
	
	Let $L_{\mathrm{EF}}\subseteq F_{\{0,1\}}$ stand for this language $L(\mathrm{EF},\{1\})$.
	That is, $1$, $1(0+0)$, $0(0+1)$ are in $L_{\mathrm{EF}}$ but $\mathbf{0}$, $0$ and $0(0+0(0))$
	are not.
\end{exa}
\begin{exa}
	\label{example-af-automaton}
	Let $L_{\mathrm{AF}}\subseteq F_{\{0,1\}}$ stand for the least language satisfying the following
	properties:
	\begin{itemize}
		\item All trees of the form $1(s)$ are members of $L_{\mathrm{AF}}$.
		\item If $n>0$ and $s=t_1+\ldots+t_n$ is a forest with $t_i\in L_{\mathrm{AF}}$ for each $i\in[n]$,
			then $s$ and $0(s)$ are members of $L_{\mathrm{AF}}$.
	\end{itemize}
	Basically, a forest belongs to $L_{\mathrm{AF}}$ iff it is nonempty and on each root-to-leaf path
	there exists a node labeled $1$.
	
	The (minimal) forest automaton of $L_{\mathrm{AF}}$ is $\mathrm{AF}=(\{0,1,2\},\{0,1\},\min,2,\cdot)$
	with $1\cdot x=1$ for each $x\in\{0,1,2\}$, $0\cdot 0=0\cdot 2=0$ and $0\cdot 1=1$.
	In $\mathrm{AF}$, a forest $s$ evaluates to $2$ if it is empty; to $1$ if it is a (nonempty) forest
	belonging to $L_{\mathrm{AF}}$; and to $0$ if it is a nonempty forest outside $L_{\mathrm{AF}}$.
\end{exa}

\subsection{Forest logics}

In this section we introduce a class of (future-time, branching) temporal logics
$\mathrm{FL}(\mathcal{L})$
(having state formulas only but no path formulas), parametrized by a set $\mathcal{L}$
of \emph{modalities}, which are forest languages themselves (not necessarily
over the same alphabet). In this section we assume that each alphabet $\Sigma$
comes with a total ordering but the expressive power of the logics will be independent
from the particular ordering chosen.

Though the \emph{syntax} of $\mathrm{FL}(\mathcal{L})$ is (essentially) the same
as the logics of~\cite{DBLP:journals/corr/abs-1208-6172}, the \emph{semantics}
is slightly different.
The change we propose in the semantics has a corollary which we find
a mathematically ``nice'' property:
in the semantics used in~\cite{DBLP:journals/corr/abs-1208-6172},
there are \emph{two} different satisfaction relations, $\models_t$ and $\models_f$
(tree and forest satisfaction, respectively), and these relations do not coincide
for trees: given a forest formula $\varphi$ and a tree $t=a(s)$, it can happen that
$t\models_t\varphi$ but not $t\models_f\varphi$ (that is, $t$ satisfies the formula
$\varphi$ viewed as a tree but not when viewed as a forest) or vice versa.
The reason is that the relation $\models_t$ automatically ``steps down'' one level
in $t$, i.e. $a(s)\models_t\varphi$ iff $s\models_f\varphi$ which is clearly different
than $t\models_f\varphi$.
The satisfaction relation of the semantics proposed in our paper is consistent
in this regard, i.e., there is no need for defining different satisfaction relations
for trees and forests: a tree satisfies a forest formula iff it satisfies the formula
viewed as a forest consisting of a single tree.

\subsubsection{Syntax.}
Given an alphabet $\Sigma$, and a class $\mathcal{L}$ of forest languages
(which are not necessarily $\Sigma$-languages), then the sets of \emph{tree formulas}
and \emph{forest formulas} of the logic $\mathrm{FL}(\mathcal{L})$ over $\Sigma$
are defined via mutual induction as the least sets satisfying all the following conditions:
\begin{itemize}
	\item $\top$ and $\bot$ are forest formulas.
	\item Each $a\in \Sigma$ is a tree formula.
	\item Every forest formula is a tree formula as well.
	\item If $\varphi$ and $\psi$ are forest formulas, then so are $(\neg\varphi)$ and
	$(\varphi\wedge\psi)$.
	\item If $\varphi$ and $\psi$ are tree formulas, then so are $(\neg\varphi)$ and
	$(\varphi\wedge\psi)$.
	\item If $L\in\mathcal{L}$ is a forest language over some alphabet $\Delta$
	and to each $\delta\in\Delta$, $\varphi_\delta$ is a tree formula over $\Sigma$,
	then $L(\varphi_\delta)_{\delta\in\Delta}$ is a forest formula (over $\Sigma$).
\end{itemize}
As usual, we use the shorthands $(\varphi\vee\psi)=\neg(\neg\varphi\wedge\neg\psi)$,
$\varphi\to\psi=\neg\varphi\vee\psi$ and remove redundant parentheses according to
the usual precedence of operators.
\begin{exa}
	\label{example-ex-syntax}
	Let $\Sigma=\{a,b,c,d\}$. Then $\varphi_0=a\vee c$ and $\varphi_1=b\vee c$ are tree
	formulas over $\Sigma$. Let $L_{\mathrm{EX}}\subseteq F_{\{0,1\}}$ be the
	language consisting of those forests having a depth-one node labeled $1$
	(e.g., $0(1(0))+0(1)$ is in $L_{\mathrm{EX}}$ but $\mathbf{0}$, $0+1$ and $0(0(1+0)+0)+0$ are not).
	Then $L_{\mathrm{EX}}(i\mapsto\varphi_i)_{i\in\{0,1\}}$ is a forest formula over $\Sigma$.
\end{exa}
\subsubsection{Semantics.}
For the semantics, tree formulas are evaluated on trees and forest formulas are
evaluated on forests. In both cases, $t\models\varphi$ denotes the fact that 
the structure $t$ (whether tree or forest) satisfies the formula $\varphi$.
This will not introduce ambiguity since although every forest formula is a tree
formula, but on trees, the two evaluation semantics coincide:
\begin{itemize}
	\item Every forest satisfies $\top$ and no forest satisfies $\bot$.
	\item The tree $b(s)$ satisfies $a\in\Sigma$ iff $a=b$.
	\item The tree $t$ satisfies the forest formula $\varphi$ if $t$, viewed as a forest consisting
	of a single tree, satisfies $\varphi$.
	\item Boolean connectives are handled as usual.
	\item The forest $s$ satisfies $L(\varphi_\delta)_{\delta\in\Delta}$ iff
	the characteristic forest of $s$ given by $(\varphi_\delta)_{\delta\in\Delta}$
	(defined below) belongs to $L$.
\end{itemize}

First we give an informal description of the characteristic forest. Given the $\Sigma$-forest $s$,
the forest $\widehat{s}$ is a $\Delta$-relabeling of $s$. For each node of $s$, one finds
the first $\delta\in\Delta$ such that the subtree of $s$ rooted at the node in question satisfies
$\varphi_\delta$; this $\delta$ is the label of the node in $\widehat{s}$.

Formally, given a forest $s$ over some alphabet $\Sigma$ and a family $(\varphi_\delta)_{\delta\in\Delta}$
of tree formulas over $\Sigma$, indexed by some alphabet $\Delta$,
we define the \emph{characteristic forest} $\widehat{s}\in F_\Delta$ of $s$
given by $(\varphi_\delta)_{\delta\in\Delta}$
by structural induction as follows: $\widehat{\mathbf{0}}=\mathbf{0}$,
$\widehat{t_1+\ldots+t_n}=\widehat{t_1}+\ldots+\widehat{t_n}$ and
$\widehat{a(s)}=b(\widehat{s})$ where $b\in\Delta$ is the \emph{first} symbol of $\Delta$ with
$a(s)\models\varphi_b$. If there is no such letter, then $b$ is the last symbol of $\Delta$.

Note that although we assume each alphabet comes with a fixed linear ordering, but this
restriction does not have any impact on the expressive power of the logics.
In fact, we can define to each $b\in\Delta$ another formula $\psi_b$ as
$\psi_b=\mathop\bigwedge\limits_{c\neq b}\neg\varphi_c$ if $b$ is the last symbol of $\Delta$ and
$\psi_b=\varphi_b\wedge\mathop\bigwedge\limits_{c<b}\neg \varphi_c$ otherwise; then,
the resulting family $(\psi_b)_{b\in\Delta}$ is \emph{deterministic} in the sense
that for any tree $t\in T_\Sigma$ there exists exactly one symbol $b\in\Delta$ with
$t\models\psi_b$, and the characteristic forests of any forest $s$
given by the two families $(\varphi_\delta)_{\delta\in\Delta}$ and
$(\psi_\delta)_{\delta\in\Delta}$ coincide.
Thus, the particular ordering of $\Delta$
is not important (we have to choose: either to syntactically restrict the allowed formulas,
or to assume an ordering of $\Delta$, or to have a family $I$ of formulas along with
a function from $P(I)\to\Delta$, or something similar to resolve ambiguities, but in all cases,
the class of definable languages is the same).

\begin{exa}
	\label{example-logic-semantics}
	Consider the forest $s$ from Example~\ref{example-forest} over $\Sigma=\{a,b,c,d\}$
	and the formula $\varphi=L_{\mathrm{EX}}(0\mapsto a\vee c,1\mapsto b\vee c)$ from Example~\ref{example-ex-syntax}. Then, a $\Sigma$-tree satisfies $\varphi_0=a\vee c$
	iff its root symbol is labeled by either $a$ or $c$; and similarly, it satisfies
	$\varphi_1=b\vee c$ if its root is labeled by either $b$ or $c$.
	Now assuming $0<1$ in the ordering of $\{0,1\}$, the characteristic forest of $s$
	defined by $(\varphi_i)_{i\in\{0,1\}}$ is
	
	{\centering
		\hfil$s=$\quad 
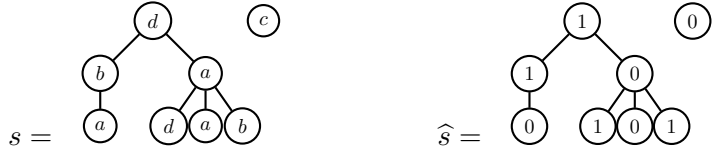
		\begin{tikzpicture}[scale=0.7,thick,every node/.style={transform shape}]
		\node[circle] (root) at (0,0) [draw] {$d$};
		\node[circle] (root2) at (2.1,0) [draw] {$c$};
		\node[circle] (node1) at (-1,-1) [draw] {$b$};
		\node[circle] (node2) at (1,-1) [draw] {$a$};
		\node[circle] (node11) at (-1,-2) [draw] {$a$};
		\node[circle] (node21) at (0.3,-2) [draw] {$d$};
		\node[circle] (node22) at (1,-2) [draw] {$a$};
		\node[circle] (node23) at (1.7,-2) [draw] {$b$};
		\path (root) edge (node1);
		\path (root) edge (node2);
		\path (node1) edge (node11);
		\path (node2) edge (node21);
		\path (node2) edge (node22);
		\path (node2) edge (node23);
		\end{tikzpicture}		
		\hfil $\widehat{s}=$\quad
		\begin{tikzpicture}[scale=0.7,thick,every node/.style={transform shape}]
		\node[circle] (root) at (0,0) [draw] {$1$};
		\node[circle] (root2) at (2.1,0) [draw] {$0$};
		\node[circle] (node1) at (-1,-1) [draw] {$1$};
		\node[circle] (node2) at (1,-1) [draw] {$0$};
		\node[circle] (node11) at (-1,-2) [draw] {$0$};
		\node[circle] (node21) at (0.3,-2) [draw] {$1$};
		\node[circle] (node22) at (1,-2) [draw] {$0$};
		\node[circle] (node23) at (1.7,-2) [draw] {$1$};
		\path (root) edge (node1);
		\path (root) edge (node2);
		\path (node1) edge (node11);
		\path (node2) edge (node21);
		\path (node2) edge (node22);
		\path (node2) edge (node23);
		\end{tikzpicture} 
		
	}
	Indeed, due to the ordering of $\{0,1\}$ nodes labeled by either $a$ or $c$ are
	relabeled to $0$; then, nodes labeled by $b$ are relabeled to $1$ since those
	subtrees satisfy $\varphi_1$; and also, nodes labeled by $d$ are also relabeled
	to $1$ since that's the last symbol of $\{0,1\}$ and the subtree does not satisfy
	either one of $\varphi_0$ or $\varphi_1$. Clearly, $L_{\mathrm{EX}}(0\mapsto a\vee c,1\mapsto\neg (a\vee c))$ is an equivalent formula.	
	Since $\widehat{s}$ is a member of $L_{\mathrm{EX}}$, we get that $s\models \varphi$.
	The language defined by $\varphi$ consists of those forests having at least one
	depth-one node labeled by a $b$ or a $d$.
\end{exa}

For a class $\mathcal{L}$ of modalities, let $\mathbf{FL}(\mathcal{L})$ denote the class
of all languages definable in $\mathrm{FL}(\mathcal{L})$.

\section{Closure properties of $\mathbf{FL}(\mathcal{L})$}
It is clear that if $K$ and $L$ are forest languages definable in $\mathrm{FL}(\mathcal{L})$ for
some $\mathcal{L}$, then so are their Boolean combinations e.g. $K\cap L$ and $\overline{K}$,
since the logic has $\wedge$ and $\neg$, thus $\mathbf{FL}(\mathcal{L})$ is closed under
(finite) Boolean combinations.

When $\varphi$ is a forest formula over the alphabet $\Sigma$ and to each $a\in\Sigma$,
$\varphi_a$ is a forest formula (also over $\Sigma$),
then we define the forest formula $\varphi[a\mapsto\varphi_a]$ inductively as
\begin{align*}
\top[a\mapsto\varphi_a] &= \top,
&\bot[a\mapsto\varphi_a] &= \bot,\\
(\neg\psi)[a\mapsto\varphi_a] &= \neg(\psi[a\mapsto\varphi_a]),
&(\psi_1\wedge\psi_2)[a\mapsto\varphi_a] &= (\psi_1[a\mapsto\varphi_a]\wedge\psi_2[a\mapsto\varphi_a]),\\
b[a\mapsto\varphi_a] &= \varphi_b,
&L(\psi_b)_{b\in\Delta}[a\mapsto\varphi_a] &= L(\psi_b[a\mapsto\varphi_a])_{b\in\Delta},
\end{align*}
that is, we replace each subformula of the form $a\in\Sigma$ of $\varphi$ by $\varphi_a$.

A \emph{literal homomorphism} of forests defined by a mapping $h:\Sigma\to\Delta$
maps a forest $s\in F_\Sigma$ to $h(s)\in F_\Delta$ given inductively as
$h(\mathbf{0})=\mathbf{0}$, $h(a(s'))=(h(a))(h(s'))$ and $h(t_1+\ldots+t_n)=h(t_1)+\ldots+h(t_n)$
(which is clearly a homomorphism). When $L\subseteq F_\Delta$ is a language
and $h:\Sigma\to\Delta$ is a mapping, then the \emph{inverse literal homomorphic image}
of $L$ is the language $h^{-1}(L)=\{s\in F_\Sigma:h(s)\in L\}$.

\begin{prop}
	$\mathbf{FL}(\mathcal{L})$ is closed under inverse literal homomorphisms.
\end{prop}
\begin{proof}
	Let $L$ be a $\Delta$-language in $\mathbf{FL}(\mathcal{L})$ defined by a formula $\varphi$ of
	$\mathrm{FL}(\mathcal{L})$ and $h:\Sigma\to\Delta$ be a mapping.
	Then $h^{-1}(L)$ is defined by $\varphi[a\mapsto\mathop\bigvee\limits_{h(b)=a}b]$
	(with the empty disjunction defined as $\bot$ of course).
\end{proof}

The following proposition states that one can use freely any definable language as a modality as well:
\begin{prop}
	\label{prop-closure}
	Assume $L$ is definable in $\mathrm{FL}(\mathcal{L})$.
	Then so is any language definable by a formula of the form $\varphi=L(\varphi_\delta)_{\delta\in\Delta}$
	with each $\varphi_\delta$ being a formula of $\mathrm{FL}(\mathcal{L})$.
\end{prop}
\begin{proof}
	We use induction on the structure of the forest formula $\psi$ defining the language $L$.
	Without loss of generality we may assume that the family $(\varphi_\delta)_{\delta\in\Delta}$
	is deterministic.
	
	If $\psi=\top$ or $\bot$, then $\varphi$ is equivalent to $\top$ or $\bot$, respectively.
	The case of the Boolean connectives $\psi=\neg\psi_1$ and $\psi=\psi_1\wedge\psi_2$
	is also clear: applying the induction hypothesis we get that there is a formula
	$\varphi_i$ of $\mathrm{FL}(\mathcal{L})$ equivalent to $\psi_i$, thus $\neg\varphi_1$,
	$\varphi_1\wedge\varphi_2$ are then formulas of $\mathrm{FL}(\mathcal{L})$, respectively,
	equivalent to $\varphi$.
	
	Finally, assume $\psi=K(\psi_\gamma)_{\gamma\in\Gamma}$ for some $\Gamma$-forest language
	$K\in\mathcal{L}$ and tree formulas $\psi_\gamma$, $\gamma\in\Gamma$ of $\mathrm{FL}(\mathcal{L})$.
	Then $K(\psi_\gamma[\delta\mapsto\varphi_\delta])_{\gamma\in\Gamma}$ is an $\mathrm{FL}(\mathcal{L})$-formula
	equivalent to $\varphi$.
\end{proof}
Since for any class $\mathcal{L}$ we also have $\mathcal{L}\subseteq\mathbf{FL}(\mathcal{L})$
(a language $L$ is defined by the formula $L(a\mapsto a)$)
and $\mathcal{L}\subseteq\mathcal{L}'$ clearly implies $\mathbf{FL}(\mathcal{L})\subseteq\mathbf{FL}(\mathcal{L}')$, along with Proposition~\ref{prop-closure}
we get that the transformation $\mathcal{L}\mapsto\mathbf{FL}(\mathcal{L})$ is a \emph{closure operator}.

\section{Facts and operations of forest automata}
Since the automaton model is complete and deterministic, thus
given $A=(Q,\Sigma,+,0,\cdot)$ over $\Sigma$ and $F\subseteq Q$, then
$L(A,F)=F_\Sigma - L(A,Q-F)$.
Also, the \emph{direct product} $A$ of the forest automata $A_i=(Q_i,\Sigma,+_i,0_i,\cdot_i)$,
$i\in I$
over the same alphabet $\Sigma$ for some index set $I$ is defined as
$\mathop\prod\limits_{i\in I}A_i=(Q,\Sigma,+,0,\cdot)$ with
$(Q,+,0)$ being the direct product of the monoids $(Q_i,\Sigma,+_i,0_i)$
and $a\cdot(q_i)_{i\in I}=(a\cdot_i q_i)_{i\in I}$ which is finite if so are 
each $A_i$ and $I$.
It is clear that if to each $i\in I$ we also have a set $F_i\subseteq Q_i$
of final states, then $A$ recognizes $\mathop\bigcap\limits_{i\in I}L(A_i,F_i)$
with the set $\mathop\prod\limits_{i\in I}F_i$ of final states. 

The automaton $A'=(Q',\Sigma,+',0,\cdot')$ is a \emph{subautomaton} of $A=(Q,\Sigma,+,0,\cdot)$
if $(Q',+',0)$ is a submonoid of $(Q,+,0)$ and $a\cdot' q=a\cdot q$
for each $a\in\Sigma$ and $q\in Q'$ (that is, $Q'\subseteq Q$ is closed under
the addition and the action, and $+'$ and $\cdot'$ are the restrictions of the
operations onto $Q'$). The \emph{connected part} of $A$ is its smallest subautomaton
(which is generated by the state $0$). An automaton is \emph{connected} if it has no
proper subautomata. Clearly, if $F\subseteq Q'$, then $L(A',F)=L(A,F)$ in this case.

The $\Delta$-automaton $A'=(Q,\Delta,+,0,\cdot')$ is a \emph{renaming} of the $\Sigma$-automaton
$A=(Q,\Sigma,+,0,\cdot)$ if for each $\delta\in\Delta$ there exists some $h(\delta)=\sigma\in\Sigma$
with $\delta\cdot'q=\sigma\cdot q$ for each state $q\in Q$.
It is straightforward to check that if $F\subseteq Q$, then $L(A',F)=h^{-1}(L(A,F))$ in this
case.

Given $A=(Q,\Sigma,+,0,\cdot)$ and $A'=(Q',\Sigma,+',0',\cdot')$,
a \emph{homomorphism} from $A$ to $A'$ is a mapping $h:Q\to Q'$ respecting the
operations: $h(0)=0'$ and $h(p+q)=h(p)+'h(q)$, $h(a\cdot q)=a\cdot'h(q)$ for each
$p,q\in Q$ and $a\in\Sigma$.
It is a routine matter to check that if in this case $F'\subseteq Q'$ and
$F=h^{-1}(F')$, then $L(A,F)=L(A',F')$.
If the homomorphism is onto, then $A'$ is a \emph{homomorphic image} of $A$,
and homomorphic images of subautomata of $A$ are called \emph{quotients} of $A$.
Clearly, the mapping $s\mapsto s^A$ is a homomorphism from $F_\Sigma$ to $A$ which is onto
if and only if $A$ is connected. 

A \emph{congruence} of $A=(Q,\Sigma,+,0,\cdot)$ is an equivalence relation $\Theta\subseteq Q^2$
such that whenever $p_1\Theta p_2$ and $q_1\Theta q_2$, then $(p_1+p_2)\Theta(q_1+q_2)$,
and $(a\cdot p_1)\Theta (a\cdot p_2)$ as well. Then, the \emph{factor automaton}
$A/\Theta=(Q/\Theta,\Sigma,+/\Theta,0/\Theta,\cdot/\Theta)$ defined by
$p/\Theta=\{q\in Q:p\Theta q\}$ standing for the class of $p$,
$X/\Theta=\{p/\Theta:p\in X\}$ for each $X\subseteq Q$,
$p/\Theta+ q/\Theta=(p+q)/\Theta$ and $a\cdot( p/\Theta)=(a\cdot p)/\Theta$
is a well-defined automaton and is a homomorpic image of $A$ via the mapping
$q\mapsto q/\Theta$.

Given any forest automaton $A=(Q,\Sigma,+,0,\cdot)$ and a subset $F\subseteq Q$
of its states, there is a \emph{minimal} forest automaton $A_L$ of the language $L=L(A,F)$,
unique up to isomorphism, which is a quotient of $A$ (and of any forest automaton
recognizing $L$). Moreover, $A_L$ can be effectively constructed from $A$ in polynomial time.

%

\section{General algebraic characterization of $\mathbf{FL}(\mathcal{L})$ by the Moore product}

In this section we show that there exists an Eilenberg-type correspondence between
language classes of the form $\mathbf{FL}(\mathcal{L})$ and pseudovarieties of finite forest
automata, closed additionally under an operation which we call the Moore product
(provided $\mathcal{L}$ satisfies a natural property).
The correspondence and the Moore product itself is the analog of 
the operation with the same name defined in~\cite{DBLP:journals/fuin/EsikI08} for
ranked trees. We think that the usage of forest automata instead of strictly ranked
universal algebras (i.e., tree automata) gives a clearer view on the connection of
the Moore product and the $L(\delta\mapsto\varphi_\delta)$-construct, defined originally
in~\cite{DBLP:journals/tcs/Esik06a} for temporal logics on \emph{trees}.

Given a forest automaton $A_1=(Q_1,\Sigma,+_1,0_1,\cdot_1)$ over some alphabet $\Sigma$ and a forest automaton $A_2=(Q_2,\Delta,+_2,0_2,\cdot_2)$ over
some alphabet $\Delta$ along with a \emph{control function} $\alpha:\Sigma\times Q_1\to\Delta$, the
\emph{Moore product} of $A_1$ and $A_2$ defined by $\alpha$ is the $\Sigma$-forest automaton $A_1\times_\alpha A_2=(Q_1\times Q_2,\Sigma,+,0,\cdot)$
with $(Q_1\times Q_2,+,0)$ being the ordinary direct product of the two horizontal monoids and
$a\cdot(p,q)=(a\cdot_1 p,b\cdot_2 q)$ with $b=\alpha(a,a\cdot_1 p)$.

For a class $\mathbf{K}$ of finite forest automata, let $\langle\mathbf{K}\rangle_M$ stand for the Moore pseudovariety of finite forest
automata generated by $\mathbf{K}$, i.e. the smallest class of finite forest automata which contains $\mathbf{K}$ and is closed under
homomorphic images, renamings, subautomata and Moore products.

Then the following holds:
\begin{prop}
	Let $\varphi=L(\varphi_\delta)_{\delta\in\Delta}$ be a formula over $\Sigma$ defining the $\Sigma$-forest language $L_\varphi$,
	with each $\varphi_\delta$, $\delta\in\Delta$
	defining the $\Sigma$-forest language $L_\delta$. Assume each $L_\delta$ is recognizable in the forest automaton $A_\delta$ and
	$L$ is recognizable in $A$. Then $L_\varphi$ is recognizable in some Moore product $A'=\bigl(\mathop\prod\limits_{\delta\in\Delta}A_\delta\bigr)\times_\alpha A$.
\end{prop}
\begin{proof}
	Let $L_\delta$ be $L(A_\delta,F_\delta)$ and $L=L(A,F)$. We define the control function $\alpha$ as
	\[
	\alpha(\sigma,(q_\delta)_{\delta\in\Delta})=\begin{cases}
	\hbox{the first }\delta\in\Delta\hbox{ such that }q_\delta\in F_\delta\hbox{ if there is such a }\delta\hbox{ at all};\\
	\hbox{the last element of }\Delta\hbox{ otherwise.}
	\end{cases}
	\]
	It is straightforward to check that for any forest $s$, the value of $s$ in this product automaton $A'$ is $((q_\delta)_{\delta\in\Delta},q)$
	where $q_\delta=s^{A_\delta}$ and $q=\hat{s}^A$ where $\hat{s}$ is the characteristic forest of $s$ with respect to the family
	$(\varphi_\delta)_{\delta\in\Delta}$, thus setting the final states to $F'=\bigl(\mathop\prod\limits_{\delta\in\Delta}Q_\delta\bigr)\times F$
	we get that $L_\varphi=L(A',F')$.
\end{proof}
For the reverse direction it suffices to show the following:
\begin{prop}
	Assume $A=(Q,\Sigma,+,0,\cdot)$ and $A'=(Q',\Delta,+',0',\cdot')$ are $\Sigma$- and $\Delta$-forest automata, respectively,
	such that every language recognizable in them is a member of $\mathbf{FL}(\mathcal{L})$ for the language class $\mathcal{L}$.
	Then every language recognizable in any Moore product of the form $A\times_\alpha A'$ is also a member of $\mathbf{FL}(\mathcal{L})$.
\end{prop}
\begin{proof}
	To each $q\in Q$ let $\varphi_q$ be the $\Sigma$-formula of $\mathrm{FL}(\mathcal{L})$ defining the language $L(A,\{q\})$,
	and to each $q'\in Q'$ let $\psi_{q'}$ be the $\Delta$-formula defining $L(A',\{q'\})$. It suffices to show that
	each language $L(A\times_\alpha A',\{(q,q')\})$ is definable in the logic by some formula $\varphi_{q,q'}$
	(since then $L(A\times_\alpha A',F)$ is definable by $\mathop\bigvee\limits_{(q,q')\in F}\varphi_{q,q'}$).
	Consider the formula $\varphi_q\wedge L_{q'}(\varphi_\delta)_{\delta\in\Delta}$ where $\varphi_\delta=\mathop\bigvee\limits_{\alpha(\sigma,p)=\delta}\sigma\wedge\varphi_p$.
	This formula defines the language $L(A\times_\alpha A',\{(q,q')\})$ and
	by Proposition~\ref{prop-closure}, there is an equivalent $\mathrm{FL}(\mathcal{L})$-formula
	since by assumption each $L_{q'}$ is definable in $\mathrm{FL}(\mathcal{L})$.
\end{proof}

Implying,
\begin{thm}
	\label{theorem-characterization-general}
	Suppose $\mathcal{L}$ is a class of regular forest languages and $\mathbf{K}$ is
	a class of forest automata such that i) each member of $\mathcal{L}$ is recognizable by
	some member of $\mathbf{K}$ and ii) every language recognizable by some member of $\mathbf{K}$
	is a member of $\mathbf{FL}(\mathcal{L})$.
	
	Then the following are equivalent to any regular forest language $L$:
	\begin{itemize}
		\item $L$ is definable in $\mathrm{FL}(\mathcal{L})$.
		\item The minimal forest automaton of $L$ belongs to $\langle\mathbf{K}\rangle_M$.
	\end{itemize}
\end{thm}

\section{Two fragments of the logic CTL}
In this section we define the two modalities of CTL we worked with, first from the 
logical perspective.
\begin{defi}
	Given a alphabet $\Sigma$, the set of TL[EF,AF]-formulas is the least set satisfying
	the following conditions:
	\begin{itemize}
		\item Each $a\in\Sigma$ is a tree formula of TL[EF,AF].
		\item Boolean combinations of tree formulas are tree formulas.
		\item $\top$ and $\bot$ are forest formulas.
		\item If $\varphi$ is a tree formula, then $\mathrm{AF}(\varphi)$
			and $\mathrm{EF}(\varphi)$ are forest formulas.
		\item Boolean combinations of forest formulas are forest formulas.
		\item Every forest formula is a tree formula as well.
	\end{itemize}
\end{defi}
The semantics of the modalities is defined as follows.
\begin{defi}
	A tree $t$ satisfies a tree formula $\varphi$ of TL[EF,AF], denoted $t\models\varphi$
	if one of the following conditions hold:
	\begin{itemize}
		\item $t=a(s)$ and $\varphi=a\in\Sigma$ for some forest $s$ and symbol $a$.
		\item $\varphi=\neg(\psi)$ and it is not the case that $t\models\psi$.
		\item $\varphi=(\psi_1\vee\psi_2)$ and either $t\models\psi_1$ or $t\models\psi_2$ (or both)
			hold.
		\item $\varphi$ is a forest formula, and $t$ (as a forest consisting of a single tree)
			satisfies $\varphi$.
	\end{itemize}
	A forest $s=t_1+\ldots+t_n$ satisfies a forest formula $\varphi$ of TL[EF,AF],
	also denoted $s\models\varphi$ if one of the following conditions hold:
	\begin{itemize}
		\item $\varphi=\top$.
		\item $\varphi=\neg(\psi)$ and it is not the case that $s\models\psi$.
		\item $\varphi=(\psi_1\vee\psi_2)$ and either $s\models\psi_1$ or $s\models\psi_2$ (or both)
hold.
		\item $\varphi=\mathrm{EF}(\psi)$ and there exists a subtree $t$ of $s$ with $t\models\psi$.
			More precisely, a forest $t_1+\ldots+t_n$ satisfies $\mathrm{EF}(\psi)$
			if there exists some $i\in[n]$ such that the tree $t_i$ satisfies $\mathrm{EF}(\varphi)$; where a tree $t=a(s)$ satisfies $\mathrm{EF}(\psi)$
			if either $t\models\psi$ or $s\models\mathrm{EF}(\psi)$ holds.
		\item In the case $\varphi=\mathrm{AF}(\psi)$, a tree $t=a(s)$ satisfies $\mathrm{AF}(\psi)$
			if either $t\models\psi$ or $s\models\mathrm{AF}(\psi)$, while a forest
			$s=t_1+\ldots+t_n$ satisfies $\mathrm{AF}(\psi)$ if $n>0$ and $t_i\models\mathrm{AF}(\psi)$ for each $i\in[n]$.
	\end{itemize}
\end{defi}
The subset of TL[EF,AF]-formulas not involving the AF modality is the set of TL[EF]-formulas,
while the subset not involving EF is the set of TL[AF]-formulas.

\section{Common Moore properties of EF and AF}
The minimal automaton of the forest language associated to the modality
EF and AF are the automata $\mathrm{EF}$ and $\mathrm{AF}$ from
Examples~\ref{example-ef-automaton} and~\ref{example-af-automaton} respectively.

Applying Theorem~\ref{theorem-characterization-general} we get the following:
\begin{thm}
	Let $L$ be a regular forest language and $A$ its minimal forest automaton. Then $L$ is definable\ldots
	\begin{itemize}
		\item \ldots in TL[EF] if and only if $A\in\langle\mathrm{EF}\rangle_M$;
		\item \ldots in TL[AF] if and only if $A\in\langle\mathrm{AF}\rangle_M$;
		\item \ldots in TL[EF+AF] if and only if $A\in\langle\mathrm{EF},\mathrm{AF}\rangle_M$.
	\end{itemize}
\end{thm}
First we list several properties of these two automata which are preserved under renamings, homomorphic images
and Moore products, thereby being necessary conditions for the automaton to be a member of the 
Moore pseudovarieties.
\begin{prop}
	Every member of $\langle\mathrm{EF},\mathrm{AF}\rangle_M$ has a horizontal monoid
	which is a semilattice, i.e. $(Q,+,0)$ satisfying $x+y=y+x$ and $x+x=x$.
\end{prop}
\begin{proof}
	It is straightforward to check that $\mathrm{EF}$ and $\mathrm{AF}$ both have a semilattice horizontal monoid
	and that these properties are preserved under renamings, quotients and Moore products.
\end{proof}
We call forest automata having a semilattice horizontal monoid \emph{semilattice automata}.
To each semilattice automaton $A=(Q,\Sigma,+,0,\cdot)$ we associate the usual partial order $\leq$
on $Q\times Q$
defined as $x\leq y~\Leftrightarrow~x=x+y$, that is, we view the semilattices as meet-semilattices.
Then, in the partially ordered set $(Q,\leq)$ the element $x+y$ is the greatest lower bound (the infimum)
of the set $\{x,y\}$.

As the semilattice ordering $\leq$ determines the addition operation $+$ completely, one can
also depict finite semilattice automata as follows: first one draws the Hasse-diagram of the
partially ordered set $(Q,\leq)$, then draws the actions as arrows, just as for ordinary automata.
Figure~\ref{fig-automata-drawings} depicts $\mathrm{EF}$ and $\mathrm{AF}$. Clearly, the unit element
of the horizontal monoid (that is, the starting state) is always the largest element of the semilattice.
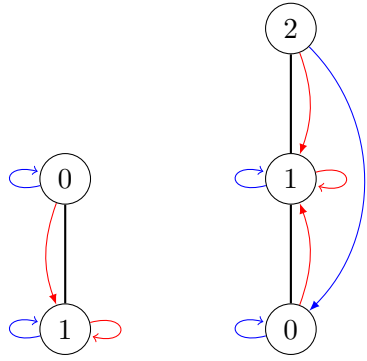
\begin{figure}[h]
	\begin{tikzpicture}
	\node[draw,circle] (1) at (0,0) {$1$};
	\node[draw,circle] (0) at (0,2) {$0$};
	\path[draw,thick] (0)--(1);
	\draw[bend right=20,-latex,red] (0) edge (1);
	\draw[loop left,-latex,blue] (0) edge (0);
	\draw[loop left,-latex,blue] (1) edge (1);
	\draw[loop right,-latex,red] (1) edge (1);
	
	\node[draw,circle] (af0) at (3,0) {$0$};
	\node[draw,circle] (af1) at (3,2) {$1$};
	\node[draw,circle] (af2) at (3,4) {$2$};
	\path[draw,thick] (af0)--(af1)--(af2);
	\draw[bend right=20,-latex,red] (af0) edge (af1);
	\draw[bend left=20,-latex,red] (af2) edge (af1);
	\draw[bend left=45,-latex,blue] (af2) edge (af0);
	\draw[loop left,-latex,blue] (af0) edge (af0);
	\draw[loop left,-latex,blue] (af1) edge (af1);
	\draw[loop right,-latex,red] (af1) edge (af1);	
	\end{tikzpicture}
	\caption{The automata $\mathrm{EF}$ and $\mathrm{AF}$. Actions for $1$ are \textcolor{red}{red}
	and actions for $0$ are \textcolor{blue}{blue} arrows.}
	\label{fig-automata-drawings}
\end{figure}

\begin{prop}
	Every member of $\langle\mathrm{EF},\mathrm{AF}\rangle_M$ is \emph{letter idempotent}:
	satisfies $aax=ax$ for each letter $a$ and state $x$.
\end{prop}
\begin{proof}
	Again, in both $\mathrm{EF}$ and $\mathrm{AF}$ the actions can be verified to be letter idempotent.
	The property is clearly preserved under taking renamings and quotients. For Moore products,
	if $A=(Q_1,\Sigma,+_1,0_1,\cdot_1)$ and $B=(Q_2,\Delta,+_2,0_2,\cdot_2)$ are letter idempotent forest automata
	and $\alpha:\Sigma\times Q_1\to \Delta$ is a control function, then in
	the Moore product $A\times_\alpha B$ we have
	\begin{align*}
	a\cdot a\cdot(p,q) &= a\cdot (a\cdot p,\alpha(a,a\cdot p)\cdot q)\\
	&= (a\cdot a\cdot p,\alpha(a,a\cdot a\cdot p)\cdot \alpha(a,a\cdot p)\cdot q)\\
	&= (a\cdot p,\alpha(a,a\cdot p)\cdot\alpha(a,a\cdot p)\cdot q)\\
	&= (a\cdot p,\alpha(a,a\cdot p)\cdot q)\\
	&= a\cdot(p,q),
	\end{align*}
	proving the claim. (We omitted the subscripts in the actions for better readability.)
\end{proof}
\section{The case of TL[EF]}
In this section we characterize $\langle\mathrm{EF}\rangle_M$.
\begin{prop}
	In each member of $\langle \mathrm{EF}\rangle_M$ we have $ax\leq x$ for each letter $a$ and state $x$.
\end{prop}
\begin{proof}
	The automaton $\mathrm{EF}$ satisfies this property as the letter $1$ maps both states to the least state $1$,
	and $0$ acts as the identity function. The property is clearly preserved under renamings and subautomata.
	For homomorphic images, if $\Theta$ is a congruence of the automaton $A=(Q,\Sigma,+,0,\cdot)$ satisfying
	the property, then for each class $x/\Theta$ we have $a\bigl(x/\Theta\bigr)+x/\Theta=(ax)/\Theta+x/\Theta=
	(ax+x)/\Theta=(ax)/\Theta=a\bigl(x/\Theta)$, thus the property is satisfied.
	
	Finally, if $A=(Q,\Sigma,+,0,\cdot)$ and $B=(Q',\Delta,+',0',\cdot')$ are forest automata satisfying the property,
	and $\alpha:\Sigma\times Q\to \Delta$ is a control function, then
	\begin{align*}
	a\cdot(x,y)+(x,y) &= (a\cdot x,\alpha(a,a\cdot x)\cdot y)+(x,y)\\
	&= (a\cdot x+x,\alpha(a,a\cdot x)\cdot y+y)\\
	&= (a\cdot x,\alpha(a,a\cdot x)\cdot y)\\
	&= a\cdot (x,y),
	\end{align*}
	thus the property is indeed preserved under taking Moore products.
\end{proof}

So weknow that each member of $\langle\mathrm{EF}\rangle_M$ is a
letter idempotent semilattice automaton
satisfying the inequality $ax\leq x$. It turns out these properties are also sufficient:
\begin{thm}
	A connected forest automaton belongs to $\langle\mathrm{EF}\rangle_M$ if and only if
	it is a letter idempotent semilattice automaton satisfying $ax\leq x$.
\end{thm}
\begin{proof}
The main idea of the proof is that whenever $A$ is a forest automaton satisfying these properties,
then $A$ belongs to the Moore pseudovariety generated by the proper homomorphic images of $A$
and the automaton $\mathrm{EF}$.

So let $A=(Q,\Sigma,+,0,\cdot)$ be a letter idempotent semilattice automaton satisfying $ax\leq x$ for each $a\in \Sigma$
and $x\in Q$. We apply induction on $|Q|$ to show that $A\in\langle\mathrm{EF}\rangle_M$.
If $|Q|=1$, then we are done since the trivial automaton belongs to any nonempty pseudovariety
so assume $|Q|>1$.

It is clear that any semilattice automaton has a least element $\mathop\sum\limits_{q\in Q}q$. Let $q_0$
denote this state of $A$. By $ax\leq x$ we get that $aq_0=q_0$ for each letter $a\in \Sigma$.

As $Q$ is a finite semilattice having at least two elements, there is at least one atom of $Q$,
that is, an element $x\neq q_0$ such that there is no $y$ with $q_0<y<x$. So let $p$ be an atom
of $Q$. Then we claim that the equivalence relation $\Theta_p$
which merges $\{q_0,p\}$ and leaves the other
states in singleton classes, is a homomorphism of $A$. Indeed, by $ap\leq p$ we get that
$ap\in\{p,q_0\}$ and $aq_0=q_0$, thus the actions are compatible with $\Theta_p$.
Moreover, for any state $q$ we have that $p+q\in\{p,q_0\}$ (as $p$ is an atom, the infimum is either $p$
or $q_0$) and $q_0+q=q_0$, thus the addition is also compatible with $\Theta_p$.
As $A/\Theta_p$ is also a letter idempotent semilattice automaton satisfying $ax\leq x$, and has $|Q|-1$ states,
applying the induction hypothesis we get $A/\Theta_p\in\langle\mathrm{EF}\rangle_M$.

Now we have two cases. Either there are at least two atoms, or there is only one.

If there are at least two atoms in $Q$, say $p$ and $q$, then $\Theta_p\cap\Theta_q$ is the identity
relation, that is, the intersection of two nontrivial congruences is the identity. Then,
$A$ divides the direct product $A/\Theta_p\times A/\Theta_q$ (it is subdirectly reducible),
which are two automata belonging to $\langle \mathrm{EF}\rangle_M$, thus $A$ is also a member of this class.

If there is exacly one atom $p$ of $Q$, then $A'=A/\Theta_p$ belongs to $\langle\mathrm{EF}\rangle_M$ by induction.
Thus, it suffices to show that $A$ belongs to $\langle A',\mathrm{EF}\rangle_M$.

The idea is the following. Let $Q'$ denote $Q/\Theta_p$.
To ease notation, we identify each $\Theta_p$-class with its least element,
i.e. the classes $\{q\}$ with $q\notin\{p,q_0\}$ with the state $q$,
and the class $\{p,q_0\}$ with $q_0$.

For a forest $s$, let $\mathrm{states}(s)$ denote the set
$\{t^{A'}:t\hbox{ is a subtree of }s\}$ of the states visited by $A'$
upon evaluating $s$. That is,
\begin{itemize}
	\item $\mathrm{states}(\mathbf{0})=\emptyset$,
	\item $\mathrm{states}(t_1+\ldots+t_n)=\mathop\bigcup\limits_{i\in[n]}\mathrm{states}(t_i)$,
	\item $\mathrm{states}(a(s))=\{(a\cdot s)^{A'}\}\cup\mathrm{states}(s)$.
\end{itemize}
For any finite set $H$, the $H$-automaton $P(H)=(P(H),H,\cup,\emptyset,\cdot)$ with
$h\cdot H'=\{h\}\cup H'$ belongs to $\langle\mathrm{EF}\rangle_M$:
first, one considers the $H$-renaming $E_h=(\{0,1\},H,\vee,0,\cdot)$ of $\mathrm{EF}$
with $h$ acting as $1$ and all other $h'\neq h$ acting as $0$, then the direct product
$\mathop\prod\limits_{h\in H}E_h$ is isomorphic to the above automaton under the
mapping $(e_h)_{h\in H}\mapsto\{h\in H:e_h=1\}$. Moreover, as $A'$ is a semilattice
automaton satisfying $ax\leq x$, we have that $s^{A'}$ is always the sum of the members
of $\mathrm{states}(s)$. Hence, the  Moore product $A'\times_\alpha P(Q')$ with
$\alpha(a,p)=p$ for each $a\in\Sigma$, $p\in Q'$ is isomorphic to
the automaton $P(A')=(P(Q'),\Sigma,\cup,\emptyset,\cdot)$ with
$a\cdot H=H\cup\{a\cdot\mathop\sum\limits_{q\in H}q\}$ for each $a\in \Sigma$ and $H\subseteq Q'$.

We now define an auxiliary automaton $\mathrm{Aux}=(\{0,1,2\},\Delta,\max,0,\cdot)$ over the
$4$-letter alphabet $\Delta7\{\ell,o,e,s\}$ so that $A$ will be a quotient of $P(A')\times_\alpha\mathrm{Aux}$ for some suitable control function $\alpha$.

Summarizing the requirements, $\mathrm{Aux}$ satisfies \ldots
\begin{itemize}
	\item $\ell\cdot 0=0$, indicating that we are not yet in
	the set $\{q_0,p\}$; for being well-defined, let $\ell\cdot 1=1$ and $\ell\cdot 2=2$;
	\item $o\cdot 0 = o\cdot 1= o \cdot 2 = 2$, indicating that we are
	already in $q_0$;
	\item $s\cdot 0=s\cdot 1=1$ and $s\cdot 2=2$, indicating that
	if we were not yet in $q_0$, then we have reached $p$ now (or sooner),
	otherwise we remain in $q_0$;
	\item $e\cdot 1=1$ and $e\cdot 0=e\cdot 2=2$, indicating that
	if we were so far in $p$, then we are still in $p$, otherwise we are in $q_0$ now.
\end{itemize}
We claim that if $\mathrm{Aux}$ is such an automaton, then for the Moore product
$\mathrm{P}=P(A')\times_\alpha\mathrm{Aux}$ with $\alpha$ given as
\begin{align*}
\alpha(a,H) &=\begin{cases}
\ell &\hbox{if }\sum H\neq q_0;\\
o &\hbox{if }\sum H=q_0\hbox{ and }a\cdot p=q_0;\\
s &\hbox{if }\sum H=q_0,~a\cdot p=p\hbox{ and }a\cdot(\sum (H-\{q_0\}))=p;\\
e &\hbox{if }\sum H=q_0,~a\cdot p=p\hbox{ and }a\cdot(\sum (H-\{q_0\}))=q_0\\
\end{cases}
\end{align*}
it holds that for any tree $t$, $t^P=(\mathrm{states}(t),x)$ with $x=0$
if $t^A\notin\{p,q_0\}$, $x=1$ if $t^A=p$ and $x=2$ if $t^A=q_0$.

As the first factor of $P$ is $P(A')$, the first entry being $\mathrm{states}(t)$
is clear. Now we proceed by induction. Let us write $t=a(s)$
with $s=t_1+\ldots +t_n$
and assume the claim holds for $t_1,\ldots,t_n$. Now if $t^A\notin\{p,q_0\}$,
then $t_i^A\notin\{p,q_0\}$ either, thus $t_i^P=(\mathrm{states}(t_i),0)$ for
each $i\in[n]$. Also, $s^A\notin\{p,q_0\}$ as well (due to $ax\leq x$),
hence $\alpha(a,H)=\ell$ for $H=\mathop\bigcup\limits_{i\in[n]}\mathrm{states}(t_i)$,
yielding $t^P=(\mathrm{states}(t),0)$ as well.

Now assume $t^A=p$. There are two cases: either $t_i^A=p$ for some $i\in[n]$,
or $p<t_i^A$ for each $i\in[n]$.
\begin{itemize}
	\item If $t_i^A=p$ for some $i\in[n]$, then by induction, $t_i^P=(H_i,1)$ and $t_j^P=(H_j,x)$
		for each $j\in[n]$ with $x\in\{0,1\}$. Thus $s^P=(H,1)$ for $H=\mathrm{states}(s)$.
		Also, in this case $s^A=p$ as well, thus $a\cdot p=p$ by letter idempotence. Hence $\alpha(a,H)$
		is either $s$ or $e$, but in both cases, $t^P=(\mathrm{states}(t),1)$.
	\item If $p<t_i^A$ for each $i\in [n]$, then $s^P=(H,0)$ for $H=\mathrm{states}(s)$.
		As $t^A=p$, we have $a\cdot p=p$. Thus, $t^P=(\{q_0\}\cup H,s\cdot 0)=
		(\mathrm{states}(t),1)$ in this case.
\end{itemize}

Finally, assume $t^A=q_0$. There are three cases: either $t_i^A=q_0$ for some
$i\in[n]$, or $q_0<t_i^A$ for each $i\in[n]$ but $t_i^A=p$ for some $i\in[n]$,
or $p<t_i^A$ for each $i\in[n]$.
\begin{itemize}
	\item If $t_i^A=q_0$ for some $i\in[n]$, then $t_i^P=(\mathrm{states}(t_i),2)$ by
	induction, thus $s^P=(H,2)$ and since each one of $o$, $s$ and $e$ map $2$ to $2$,
	we get that $t^P=(\mathrm{states}(t),2)$ (as $\mathrm{states}(t)$ contains $q_0$).
	\item If $t_i^A=p$ for some $i\in[n]$ and $p\leq t_j^A$ for each $j\in[n]$,
	then $s^A=p$ as well since $p$ is the only atom of $A$. Then $a\cdot p=q_0$,
	thus $\alpha(a,\mathrm{states}(t))=o$ and hence $t^P=(\mathrm{states}(t),2)$.
	\item Finally, if $p<t^A_i$ for each $i\in[n]$, then by the induction hypothesis
	$t_i^P=(\mathrm{states}(t_i),0)$ and $s^P=(\mathrm{states}(s),0)$.
	Moreover, for $t^P=(H,x)$ we have that $H-\{q_0\}=\mathrm{states}(s)$,
	and $\sum(H-\{q_0\})=s^A$. Hence $a\cdot(\sum(H-\{q_0\}))=q_0$,
	and $\alpha(a,\mathrm{states}(t))=e$. By $e(0)=2$ we get $t^P=(\mathrm{states}(t),2)$.
\end{itemize}
It remains to show that there exists such an automaton $\mathrm{Aux}$ in
$\langle \mathrm{EF}\rangle_M$. Let 
$B=(\{0,1\},\Delta,\vee,0,\cdot)$ be the $\Delta$-renaming of $\mathrm{EF}$ with $\ell^B=e^B=0^{\mathrm{EF}}$ and $s^B=o^B=1^{\mathrm{EF}}$.
Furthermore, let $\alpha:\Delta\times\{0,1\}\to \{0,1\}$ be the mapping
\[\alpha(a,\delta)=\begin{cases}1&\hbox{if }\delta=o\textrm{ or both }\delta=e\textrm{ and }a=0;\\0&\hbox{otherwise.}\end{cases}\]

We claim that $\mathrm{Aux}$ is the homomorphic image of $B\times_\alpha\mathrm{EF}$ under the mapping $(0,0)\mapsto 0$,
$(1,0)\mapsto 1$ and $(0,1),(1,1)\mapsto 2$. It is clear that this mapping is a homomorphism between the horizontal monoids.
For the actions, consulting the following table
\begin{center}
	\begin{tabular}{|l|c|c|c|c|}
		\hline
		&$(0,0)$&$(0,1)$&$(1,0)$&$(1,1)$\\\hline
		$\ell$&$(0,0)$&$(0,1)$&$(1,0)$&$(1,1)$\\\hline
		$s$&$(1,0)$&$(1,1)$&$(1,0)$&$(1,1)$\\\hline
		$e$&$(0,1)$&$(0,1)$&$(1,0)$&$(1,1)$\\\hline
		$o$&$(1,1)$&$(1,1)$&$(1,1)$&$(1,1)$\\\hline
	\end{tabular}
\end{center}
we get that $\mathrm{Aux}$ is indeed a homomorphic image of $B\times_\alpha\mathrm{EF}$, thus $\mathrm{Aux}$ is in $\langle\mathrm{EF}\rangle_M$, hence so is $A$.
\end{proof}

\section{The case of TL[AF]}
\label{sec-af}

Let us call a forest automaton $A=(Q,\Sigma,+,0,\cdot)$ \emph{positive} if
for any forest $s$, $s^A=0$ if and only if $s=\mathbf{0}$. Then, $\mathrm{AF}$
is a positive automaton. Any connected positive automaton $A=(Q,\Sigma,+,0,\cdot)$
can be written as $A=(Q'\cup\{0\},\Sigma,+,0,\cdot)$ where $(Q',+)$ is a semigroup
(that is, closed under $+$) and the actions also map $Q'$ into itself.
Let us call this set $Q'$ the \emph{core} of $A$, denoted $\mathrm{core}(A)$.

The following is easy to see.
\begin{prop}
	Any renaming and Moore product of a positive forest automaton is positive.
	Also, if $A$ and $B$ are positive, then
	$\mathrm{core}(A\times_\alpha B)\subseteq\mathrm{core}(A)\times\mathrm{core}(B)$
	for any Moore product $A\times_\alpha B$.
\end{prop}
The next property is a bit more involved to check:
\begin{prop}
	The connected part $A$ of each nontrivial member of $\langle\mathrm{AF}\rangle_M$ is a positive automaton
	satisfying $p\leq ap$ for each $p\in\mathrm{core}(A)$.
\end{prop}
\begin{proof}
	The property holds for $\mathrm{AF}$ and is clearly preserved under renamings and subautomata.
	For Moore products of the form $C=A\times_\alpha B$ with $A,B\in\langle\mathrm{AF}\rangle_M$,
	if either $A$ or $B$ is trivial, then $C$ is a renaming of the other one. So assume $A$ and
	$B$ are both nontrivial. By induction we have that $A$ and $B$ are positive
	(thus so is $C$) and the states $p$ in their cores satisfy $p\leq ap$.
	
	Then for any letter $a\in\Sigma$ we have
	\begin{align*}
	(p,q)+a\cdot(p,q) &= (p,q)+(a\cdot p,\alpha(a,a\cdot p)\cdot q)\\
	&= (p+a\cdot p,q+\alpha(a,a\cdot p)\cdot q)\\
	&= (p, q)
	\end{align*}
	and the claim holds.
	
	Finally, let $A\in\langle\mathrm{AF}\rangle_M$ and $\Theta$ a congruence of $A=(Q,\Sigma,+,0,\cdot)$
	such that $A/\Theta$ is nontrivial. Then so is $A$, hence by induction $A$ is positive and
	each state $p$ in its core satisfy $p\leq ap$ for each $a\in\Sigma$. First we show that 
	$A/\Theta$ is positive, that is, $\{0\}$ is a singleton $\Theta$-class. Assume to the contrary
	that $q\Theta 0$ for some state $q\neq 0$. Then $q$ belongs to the core of $A$.
	Also, then for each state $p$ with $q\leq p$ we have $q=(q+p)\Theta(0+p)=p$, hence $p\Theta q$ as well.
	Applying $q\leq aq$ we get that $q/\Theta=aq/\Theta$ for each $a\in \Sigma$. Thus,
	$a\cdot 0/\Theta=a\cdot q/\Theta=q/\Theta=0/\Theta$, and by idempotence we get that every forest $s$
	evaluates to $0/\Theta$ in $A/\Theta$. Hence, $A/\Theta$ is trivial, a contradiction.
	
	Thus, if $p/\Theta$ is in the core of $A/\Theta$, then $p$ is in the core of $A$, thus
	$p\leq ap$, implying $p/\Theta\leq ap/\Theta$.
\end{proof}

We also know where the actions should map the starting state.
\begin{prop}
	For each connected $A=(Q,\Sigma,+,0,\cdot)\in\langle\mathrm{AF}\rangle_M$ and $a\in \Sigma$ it holds that
	$a\cdot 0=a\cdot\bot_A$ where $\bot_A=\sum Q$ is the least state of $A$.
\end{prop}
\begin{proof}
	The claim holds for $\mathrm{AF}$ and for any connected part of any renaming of $\mathrm{AF}$.
	Now we use the fact that every connected member of $\langle\mathrm{AF}\rangle_M$ is a homomorphic image
	of a product of the form $(\ldots((\mathrm{AF}'\times_{\alpha_1} \mathrm{AF})\times_{\alpha_2}\mathrm{AF})\times_{\alpha_3}\ldots)\times_{\alpha_n}\mathrm{AF}$ for some
	Moore product with $\mathrm{AF}'$ being a renaming of $\mathrm{AF}$ where after each $\alpha_i$-product we take
	immediately the connected part of the result.
	
	So if $A=(Q,\Sigma,+,0,\cdot)$ satisfies the property and $B$ is the connected part of 
	some Moore product $A\times_\alpha\mathrm{AF}$, then $\bot_B$
	is either $(\bot_A,0)$ (if there is some state in the connected part of the product of the form
	$(p,0)$ or $(\bot_A,1)$ (otherwise). If it is $(\bot_A,1)$, then $B$ is isomorphic to $A$ and
	the claim holds. If it's $(\bot_A,0)$, then for any $a\in\Sigma$ we have
	\begin{align*}
	a\cdot(0,2) &= (a\cdot 0,\alpha(a,a\cdot 0)\cdot 2)\\
	&= (a\cdot \bot_A,\alpha(a,a\cdot\bot_A)\cdot 2)\\
	&= (a\cdot \bot_A,\alpha(a,a\cdot\bot_A)\cdot 0)\\
	&= a\cdot(\bot_A,0)
	\end{align*}
	and thus the claim holds for Moore products.
	
	For homomorphic images, $a\cdot 0/\Theta=(a\cdot 0)/\Theta=(a\cdot\bot_A)/\Theta=a\cdot \bot_A/\Theta$	
	and of course $\bot_A/\Theta$ is the least state of $A/\Theta$, so the claim holds.
\end{proof}

The states of the core also satisfy an additional implication:
\begin{prop}
	For any nontrivial member $A$ of $\langle\mathrm{AF}\rangle_M$,
	it holds that if $p$ and $q$ are in the core of $A$
	and $a\in\Sigma$ is a letter with $p\leq q\leq ap$, then $ap=aq$.
\end{prop}
\begin{proof}
	It is straightforward to verify that the properties hold for $\mathrm{AF}$ and is clear that
	are preserved under renamings and subautomata. For Moore products $C=A\times_\alpha B$,
	if $(p,p')$ and $(q,q')$ are in the core of $C$, then $p,q$ are in the core of $A$ and
	$p',q'$ are in the core of $B$. Now assuming $(p,p')\leq (q,q')\leq a(p,p')$
	we get $p\leq q\leq ap$, implying $ap=aq$, and $p'\leq q'\leq \alpha(a,ap)p'$,
	implying $\alpha(a,ap)p'=\alpha(a,ap)q'=\alpha(a,aq)q'$.
	Hence we have $a(q,q')=(aq,\alpha(a,aq)q')=(ap,\alpha(a,ap)p')=a(p,p')$.
	
	Finally, for homomorphic images, let $A$ satisfy this property and let $\Theta$ be a congruence of $A$.
	Let $p/\Theta\leq q/\Theta\leq ap/\Theta$. We define two sequences $p_0,p_1,\ldots$ and $q_0,q_1,\ldots$
	as follows:
	\begin{itemize}
		\item Let $q_0=q$ and $p_0=p+q$.
		\item For each $n>0$, let $q_n=q_{n-1}+ap_{n-1}$ and $p_n=p_{n-1}+q_n$.
	\end{itemize}
	Then for each $n\geq 0$ we have that $p_n\Theta p$, $q_n\Theta q$,
	$p_n\leq q_n$ and $q_{n+1}\leq ap_n$. Moreover, $p_{n+1}\leq p_n$ and
	$q_{n+1}\leq q_n$. Since $A$ is finite, so is each $\Theta$-class, thus
	for some $m$ we have $p_m=p_{m+1}$ and $q_m=q_{m+1}$, yielding
	$p_m\leq q_m\leq ap_m$, hence $aq_m=ap_m$, thus $a(q/\Theta)=a(p/\Theta)$
	and the property is thus verified.
\end{proof}
We actually conjecture that the properties we enlisted so far are also
sufficient for membership in $\langle\mathrm{AF}\rangle_M$.
\begin{conj}
\label{conj-a}
	A nontrivial connected forest automaton belongs to $\langle\mathrm{AF}\rangle_M$
	if and only if it is a positive, letter idempotent semilattice automaton,
	with its states in its core satisfying $x\leq ax$
	and the implication $x\leq y\leq ax\Rightarrow ay=ax$.
\end{conj}
We have generated numerous members of $\langle\mathrm{AF}\rangle_M$ and
we were always able to find a particular type of congruence which we call
a ``ladder congruence'':
\begin{defi}
	Given a (positive, semilattice, letter idempotent) forest automaton $A$, a congruence
	$\Theta$ of $A$ is called a ``ladder congruence'' if it satisfies all the
	following conditions:
	\begin{itemize}
		\item Each $\Theta$-class consists of either one or two elements. Hence if
			$\{p,q\}$ is a class for $p\neq q$, then as $p+q$ also belongs to the class,
			it has to be the case that $p+q\in\{p,q\}$, thus either $p\leq q$ or
			$q\leq p$ holds. Moreover, there is no $r$ with $p<r<q$, since in that case
			$r$ should also belong to this $\Theta$-class as $p=(p+r)\Theta (q+r)=r$.
		\item For any $\Theta$-class $C=\{p,q\}$ consisting of two states $p<q$
			and for any letter $a$, it is either the case that $p$ is not in the
			image of $a$, or for every other  $\Theta$-class $D$ with
			$aD=C$, either $D=\{r\}$ is a singleton class and $ar=p$, or
			$D=\{r,s\}$ consists of two states $r<s$ and $ar=p$, $as=q$.
	\end{itemize}
\end{defi}
It is relatively easy to check that if $\Theta$ is a ladder congruence of
$A$, then $A$ is a quotient of $A/\Theta\times_\alpha\mathrm{AF}$ for a suitable
control function $\alpha$.

Based on our experiments, we also propose the following conjecture:
\begin{conj}
\label{conj-b}
	Every nontrivial connected member of $\langle\mathrm{AF}\rangle_M$ is either
	subdirectly reducible, or its least nontrivial congruence is a ladder
	congruence.
\end{conj}
Both of Conjectures~\ref{conj-a} and~\ref{conj-b} would imply decidability
of the membership problem of the class $\langle\mathrm{AF}\rangle_M$, thus
the decidability of the definability problem of TL[AF].

\section{Conclusions}
We defined the Moore product of forest automata and showed that this product operation corresponds exactly
to the application of temporal logic modalities on forests for a semantics slightly different from the
already existing one in the literature.
We think that characterizing the logic CTL (say) for forests might be a slightly easier research objective
than for doing the same for the setting of strictly ranked trees, and still the results might be easy to lift
to that setting as well. We think that a way seeking for decidable characterizations is to find first
several identities that hold for the algebraic bases of the logic in question, and are preserved in
Moore products, quotients and renamings. Such properties give necessary conditions for an automaton
to be a member of the corresponding pseudovariety. Then, if the set of identities is complete, one
can show that it is sufficient, either by decomposing directly using the algebraic framework
(as it's done in this paper) or by writing formulas defining the languages recognized
in the states of the automaton and it is a matter of personal taste which option one chooses
for this second direction.

Or course a decidable characterization of the full CTL logic would be very interesting to get.
It would be also interesting to know whether the identities of Section~\ref{sec-af} are complete for TL[AF].

\section*{Acknowledgements}
The authors thank Andreas Krebs and the late Zoltán Ésik for discussion on the topic.

\bibliographystyle{plain}
\bibliography{biblio}{}

\end{document}